\documentclass[%
 reprint, 
 superscriptaddress,
 amsmath,amssymb,
 aps,
 pra,
]{revtex4-2}
\pdfoutput=1 %

\usepackage{braket}
\usepackage{amsmath,amssymb,amsfonts}
\usepackage{amsthm}
\usepackage{mathtools}
\usepackage{graphicx}
\usepackage{xcolor}
\usepackage{framed}
\usepackage[colorlinks=true,
            allcolors=blue]{hyperref}

\usepackage[capitalize,nameinlink]{cleveref}
\usepackage{comment}

\newcommand{\be}{\begin{equation}}\newcommand{\ee}{\end{equation}}\newcommand{\ba}{\begin{eqnarray}}\newcommand{\ea}{\end{eqnarray}}\newcommand{\ban}{\begin{eqnarray*}}\newcommand{\ean}{\end{eqnarray*}}

\newtheorem{theorem}{Theorem} %
\newtheorem{lemma}{Lemma}

\theoremstyle{definition}

\theoremstyle{remark}

\def\mE{\mathcal{E}}

\def\mL{\mathcal{L}}
\def\mS{\mathcal{S}}

\def\openone{\mathds{1}}

\newcommand{\map}[1]{\mathcal{#1}}

\renewcommand{\set}[1]{\mathcal{#1}}

\newcommand{\ketbra}[2]{|#1\rangle\!\!\!\; \langle #2 |}

\newcommand{\tr}[1]{\operatorname{Tr}\!\left[#1\right]}

\newcommand{\HS}{{\mathcal{H}}}

\usepackage{dsfont}
\usepackage{enumitem}

\newcommand{\Tr}{\mathrm{Tr}}

\newcommand{\mR}{\mathcal{R}}

\begin{document}

\title{Quantum Bayes' rule and Petz transpose map from the minimum change principle}

\author{Ge Bai}
\affiliation{Centre for Quantum Technologies, National University of Singapore, 3 Science Drive 2, Singapore 117543}
\affiliation{Thrust of Artificial Intelligence, Information Hub, The Hong Kong University of Science and Technology (Guangzhou), Guangzhou 511453, China}
\email{gebai@hkust-gz.edu.cn}

\author{Francesco Buscemi}%
\affiliation{Department of Mathematical Informatics, Nagoya University, Furo-cho Chikusa-ku, Nagoya 464-8601, Japan}
\email{buscemi@nagoya-u.jp}

\author{Valerio Scarani}%
\affiliation{Centre for Quantum Technologies, National University of Singapore, 3 Science Drive 2, Singapore 117543}
\affiliation{Department of Physics, National University of Singapore, 2 Science Drive 3, Singapore 117542}
\email{physv@nus.edu.sg}

\begin{abstract}
    Bayes' rule, which is routinely used to update beliefs based on new evidence, can be derived from a principle of minimum change. This principle states that updated beliefs must be consistent with new data, while deviating minimally from the prior belief. Here, we introduce a quantum analog of the minimum change principle and use it to derive a quantum Bayes' rule by minimizing the change between two quantum input-output processes, not just their marginals. This is analogous to the classical case, where Bayes' rule is obtained by minimizing several distances between the joint input-output distributions. When the change maximizes the fidelity, the quantum minimum change principle has a unique solution, and the resulting quantum Bayes' rule recovers the Petz transpose map in many cases.
\end{abstract}
\maketitle

{\em Introduction.}---Usually demonstrated by simple counting arguments involving urns and balls, Bayes' rule has actually been argued to play a much deeper role in probability theory and logic, as the \textit{only consistent system} for updating one's beliefs in light of new evidence~\cite{jeffreys1998theory, polya1954induction, cox1961algebra,jeffrey1965logic, definetti1974theory, jaynes2003probability}. As an alternative to the above axiomatic approach, Bayes' rule can also be derived from a variational argument: the updated belief should be consistent with the new observations while deviating as little as possible from the initial belief. This is known as the \textit{minimum change principle}~\cite{aitchison1975goodness,may1976toward,williams1980bayesian,zellner1988optimal}. %
It formalizes the intuition that the new information should be incorporated into the agent's knowledge in the ``least committal'' way, e.g. without introducing biases unwarranted by the data. Such fundamental insights can be seen as at least a motivation, if not an explanation, for the extraordinary effectiveness of Bayesian statistical inference in virtually all areas of knowledge.

If one considers quantum theory as a noncommutative extension of probability theory, one would expect that there would also be a sound analog of Bayes' rule. However, the status of a \textit{quantum Bayes' rule} is still much debated, with many alternatives, often inequivalent to each other, having been proposed in the past decades~\cite{caves1986quantum,korotkov1999continuous,korotkov2001selective,gardiner2004quantum,luders2006concerning,ozawa1997quantum,fuchs2001quantum,schack2001quantum,warmuth2005bayes,leifer2013towards,surace2022state,tsang2022generalized,parzygnat2022non,parzygnat2023time,parzygnat2023axioms}. Among these proposals, the \textit{Petz transpose map}~\cite{petz1,petz} stands out as the only quantum Bayes' rule that satisfies a set of axioms analogous to the classical Bayes' rule~\cite{parzygnat2023axioms}.

Attempts have also been made to derive a quantum Bayes' rule from optimizations involving the ``posterior'' state. For instance, Ref.~\cite{warmuth2005bayes} minimizes a loss function related to the quantum relative entropy, while Refs.~\cite{tsang2023operational,tsang2022generalized} optimize an upper bound on a distance measure between two estimators of the initial and final states. However, these methods, while involving optimization, focus primarily on the marginals of the process rather than on the entire process. As a result, they do not fully reflect the minimality of the change \textit{over the entire process}, which is instead the core argument from which Bayes' rule and its generalizations, such as Jeffrey's theory of probability kinematics, emerge.

Thus, the current situation is that the analog of Bayes' rule in quantum theory is not yet settled, despite the importance that such a concept would have not only for the foundations of the theory, but also for its applications. %

In this work, we take a decisive step towards solving this problem by presenting a new approach to the quantum Bayes' rule, based on a natural quantum analog of the minimum change principle, which involves the entire process, not just its marginals (see Fig.~\ref{fig:useless} for a schematic representation). Specifically, when the change maximizes the quantum fidelity~\cite{jozsa1994fidelity,uhlmann1976transition}, the resulting quantum Bayes' rule can be derived analytically and corresponds to the Petz transpose map in many cases. This connection further strengthens the link between Bayes' rule, the minimum change principle, and the Petz transpose map, thus justifying their broader applications in quantum information theory and possibly beyond.

\begin{figure}
    \centering
    \includegraphics[width=0.8\linewidth]{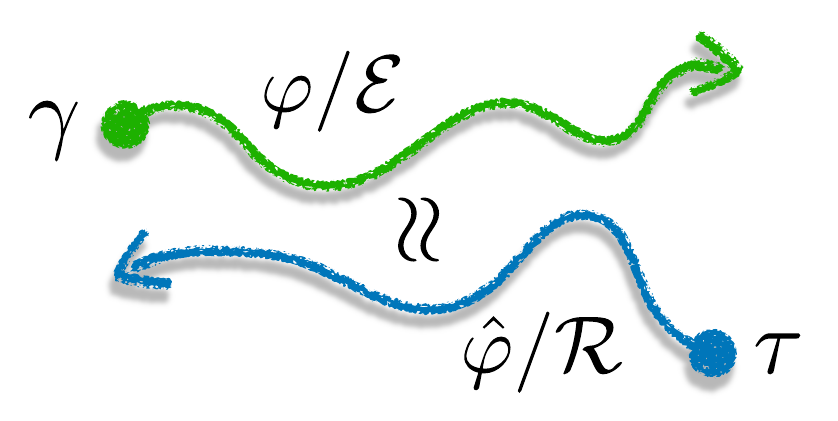}%
    \caption{Visualization of the optimization performed in this paper. Given the forward channel ($\varphi$ classical, $\mathcal{E}$ quantum) and the initial states $\gamma$ and $\tau$ of the forward and backward processes, one searches for the backward map ($\hat{\varphi}$ classical, $\mathcal{R}$ quantum) that minimizes the change or, equivalently, maximizes the similarity between the two processes. In the classical case, the processes are represented by joint probability distributions, and $\hat{\varphi}$ is known to be given by Bayes' rule for many quantitative definitions of similarity. By exploiting a suitable representation of quantum processes as bipartite quantum states, this paper presents the first result of this kind for quantum information, using the fidelity~\eqref{eq:fidelity} as the measure of similarity to maximize.}
    \label{fig:useless}
\end{figure}

\medskip{\em Classical Bayes' rule from the minimum change principle.}---Consider two random variables $X$ and $Y$ taking values in two finite alphabets, $\set{X}=\{x\}$ and $\set{Y}=\{y\}$, respectively. Assume that the initial belief about $X$ is represented by a distribution $\gamma(x)$, while the correlation between $X$ and $Y$ is given by a conditional probability distribution $\varphi(y|x)$. We can think of the latter as the ``forward process'' from $X$ to $Y$. Now, suppose we observe a certain value $Y=y_0$ and want to update our belief about $X$ in light of this new information. Bayes' rule provides the formula for the posterior probability as:
\begin{align}\label{eq:bayes}
    \hat\varphi_\gamma(x|y_0) :=  \frac{\varphi(y_0|x)\gamma(x)}{\sum_{x'\in\set{X}} \varphi(y_0|x')\gamma(x')}\;.
\end{align}

But how should the above be generalized to the situation where the new information does not come in the form of a definite value $y_0$ for $Y$, but as ``soft evidence,'' i.e., a probability distribution $\tau(y)$? As Jeffrey~\cite{jeffrey} argues, this is in fact the rule rather than the exception, since it is very rare that we can ascribe probability one to only one event $Y=y_0$, ruling out all others with absolute certainty. Jeffrey's probability kinematics~\cite{jeffrey} and Pearl's virtual evidence method~\cite{pearl} agree~\cite{CHAN200567} that the correct way to proceed is to promote Bayes' posterior to a full-fledged ``reverse process'' $\hat\varphi_\gamma(x|y)$ that yields the probabilities of $X$ given $Y$, so that the updated belief given the new evidence $\tau(y)$ becomes
\begin{align}\label{eq:jeffrey}
   \gamma'(x):= \sum_{y\in\set{Y}}\hat\varphi_\gamma(x|y)\tau(y)\;.
\end{align}
The conventional Bayes' rule~\eqref{eq:bayes} is clearly recovered as a special case of~\cref{eq:jeffrey} in which $\tau(y)$ is the delta function $\delta_{y,y_0}$. 

While Jeffrey's and Pearl's derivations of~\cref{eq:jeffrey} rely on certain natural assumptions about the logic model or underlying Bayesian network, an alternative way to obtain \cref{eq:jeffrey} is through the minimum change principle~\cite{may1976toward,aitchison1975goodness,williams1980bayesian,zellner1988optimal}. To use this principle, we consider the joint distribution of $X$ and $Y$ in the forward process, denoted as $P_{\rm fwd}(x,y):=\varphi(y|x)\gamma(x)$, and try to update it based on a subsequent observation of $Y$. The comparison will be done not only on the marginal distribution on $X$, but on the \textit{whole} joint distribution $P_{\rm fwd}(x,y)$. The minimum change principle then determines the updated belief about the joint distribution, denoted as $P_{\rm rev}(x,y)$.

Concretely, the minimum change principle can be described as the minimization of a divergence measure between the prior belief $P_{\rm fwd}$ and the updated belief $P_{\rm rev}$ under the constraint that $P_{\rm rev}$ is compatible with the new observation $\tau(y)$. This can be expressed as the following optimization problem:
\begin{align} \label{eq:clax_optimize}
\begin{split}
    \min_{P_{\rm rev}}\quad&   \mathbb{D}\left(P_{\rm fwd}, P_{\rm rev}\right) \\
    \text{subject to}\quad& \sum_{x\in\set{X}} P_{\rm rev}(x,y) = \tau(y)\;,\quad\forall y\in\set{Y}\;, \\
    \quad& P_{\rm rev}(x,y)\geq 0\;,\quad \forall x\in\set{X},~y\in\set{Y}\;,
\end{split}
\end{align}
where $\mathbb{D}\left(P_{\rm fwd}, P_{\rm rev}\right)$ is a statistical divergence measure between the prior and posterior joint distributions. Common choices include the Kullback--Leiber divergence \cite{aitchison1975goodness,zellner1988optimal}, other $f$-divergences %
including Pearson divergence and Hellinger distance \cite{liese2006divergences}, zero-one loss \cite{domingos1997optimality}, or the mean-square error of an estimation \cite{robert2007bayesian,berger2013statistical}. 
In all these cases, the optimal value is achieved by %
\begin{align}
    P_{\rm rev}(x,y)=\frac{P_{\rm fwd}(x,y)}{\sum_{x'} P_{\rm fwd}(x',y)} \tau(y) = \hat\varphi_\gamma(x|y) \tau(y)\;.
\end{align}
The above equation shows that the minimum change principle applied to the joint distributions of $X$ and $Y$ leads to the Bayes--Jeffrey rule in \cref{eq:bayes,eq:jeffrey}. 

To pave the way for the quantum case to be presented in what follows, we rephrase the optimization problem from~\cref{eq:clax_optimize} into an equivalent form as
\begin{align} \label{eq:clax_opt_gamma_tau}
\begin{split}
    \min_{\hat\varphi}\quad&  \mathbb{D}\Big(\varphi(y|x)\gamma(x) , \hat\varphi(x|y) \tau(y)\Big) \\
    \text{subject to}\quad& \sum_{x\in\set{X}} \hat\varphi(x|y) = 1\;,\quad\forall y\in\set{Y}\;,\\
    \quad& \hat\varphi(x|y)\geq 0\;,\quad \forall x\in\set{X},~y\in\set{Y}\;.
\end{split}
\end{align}
This form of optimization shifts our focus from $P_{\rm rev}$ to $\hat\varphi$, explicitly constraining $\hat\varphi$ to be a stochastic process. For the divergences mentioned above, the optimal value is still obtained with Bayes' rule~\cref{eq:bayes}. Since we are now interested in the reverse process $\hat\varphi$, the distribution $\tau$ can be seen in this context as a reference distribution of the variable $Y$ instead of the observation on $Y$. Interestingly, the solution provided by Bayes' rule is uniformly optimal regardless of the choice of $\tau$.

The goal of this work is to find a quantum analog of Bayes' rule, starting from an analog of the optimization problem~\eqref{eq:clax_optimize}. To do this, we first need to clarify the quantum equivalents of classical concepts such as stochastic processes, joint distributions, and statistical divergences.

\medskip{\em Representation of a quantum process.}---Consider two quantum systems $A$ and $B$, associated with two finite dimensional Hilbert spaces $\HS_A$ and $\HS_B$, respectively. Let $\mL_+(\HS)$ be the set of positive semidefinite operators on $\HS$ and $\mS(\HS)\subset\mL_+(\HS)$ be the subset of normalized states, i.e., density matrices or, equivalently, positive semidefinite operators $\rho\ge 0$ with unit trace. Note that strict positivity, i.e.~$\rho>0$, implies that the state $\rho$ is full-rank.

Stochastic processes in the quantum case correspond to completely positive trace-preserving (CPTP) linear maps, also known as quantum channels.
Following the classical Bayes' rule, we start with a forward quantum process $\map{E}$, i.e.\ a channel from system $A$ to system $B$. The goal is to determine a backward quantum process $\map{R}$, which is again a channel, but going from $B$ to $A$. Note that the adjoint map $\map{E}^\dag$, defined by trace duality as $\tr{\map{E}(X)\;Y}=:\tr{X\;\map{E}^\dag(Y)}$ for all operators $X$ and $Y$, is a completely positive linear map from $B$ to $A$, although it is generally not trace-preserving and thus not a channel. As in the classical case, we assume a prior belief about the input, but instead of a probability distribution, the prior belief is now given as a quantum state $\gamma\in \mS(\HS_A)$
\footnote{{There are also proposals of quantum Bayes' rule where the belief is described by objects other than a density matrix, such as a distribution over states \cite{guevara2015quantum}, or a joint state over the system of interest and a reference system \cite{liu2025state}. We leave their relation to the minimum change principle open for future studies.}}.

{Next, we look for a quantum analog of the joint input-output distribution of a stochastic process, such as the classical $P_{\rm fwd}(x,y) = \varphi(y|x)\gamma(x)$, which explicitly encodes correlations between inputs and outputs. In the classical case, the joint input-output distribution is defined directly by conditional probabilities, whereas in the quantum case, channels map input states to output states without providing a direct analog of a joint distribution.

To overcome this limitation and to represent the input-output correlations of a quantum channel, we adopt a purification-based construction. Specifically, we consider the canonical purification of $\gamma$~\cite{nielsen2010quantum,dutta2021canonical}
\begin{align}
    |\sqrt{\gamma}\rangle\!\rangle := \sum_{ij} \bra{i}\sqrt{\gamma}\ket{j} \ket{i}_{A_1}\ket{j}_{A_2}
\end{align}
where $A_1$ and $A_2$ are two copies of system $A$ and $\{\ket{i}\}$ is a chosen orthonormal basis of $\HS_A$. This pure state has the following marginal states
\begin{align}
    \Tr_{A_2}[|\sqrt{\gamma}\rangle\!\rangle\langle\!\langle\sqrt{\gamma}|]=\gamma\,, \qquad  
    \Tr_{A_1}[|\sqrt{\gamma}\rangle\!\rangle\langle\!\langle\sqrt{\gamma}|]=\gamma^T \,,
\end{align}
where the superscript notation $\bullet^T$ denotes the transposition taken with respect to the orthonormal basis $\{\ket{i}\}$.

Now, since we have two (correlated) copies of the input system, one could apply $\mathcal{E}$ to the first system and obtain $\mathcal{E}(\gamma)$, while the second system remains in state $\gamma^T$. Explicitly, this results in the state
\begin{align} \label{eq:apply_on_purification}
\begin{split}
    Q_{\rm fwd}&:=(\mathcal{E}\otimes\mathcal{I})\left(|\sqrt{\gamma}\rangle\!\rangle\langle\!\langle\sqrt{\gamma}|\right)\\
    &= \left(\openone_B \otimes \sqrt{\gamma^T}\right)C_{\map{E}}\left(\openone_B \otimes \sqrt{\gamma^T}\right),
\end{split}
\end{align}
where $\mathcal{I}$ is the identity map on system $A_2$, $\openone_B$ is the identity operator on system $B$ and $C_{\map{E}} = \sum_{i,j}\map{E}(\ketbra{i}{j}) \otimes \ketbra{i}{j} \in \mL_+(\HS_B\otimes\HS_A)$ is the Choi operator of the forward process $\map{E}$~\cite{choi1975completely}.
Computing the marginal operators we obtain
\begin{align} \label{eq:Q_F_marginals}
    \Tr_A[Q_{\rm fwd}] = \map{E}(\gamma)\;,\qquad    \Tr_B[Q_{\rm fwd}] = \gamma^T\;.
\end{align}

Such a representation of a quantum process has appeared in the literature  in a context comparing quantum processes to classical conditional probabilities~\cite{leifer2006quantum,CHRUSCINSKI2020115}. A benefit of this representation is that $Q_{\rm fwd}$ is always in $\mS(\HS_B\otimes\HS_A)$, i.e., it is a bipartite state for all choices of channel $\mE$ and prior state $\gamma$, and many divergence functions are well-defined and possess nice properties when evaluated on states.

{Although the operator $Q_{\rm fwd}$ is not a quantum state over time (the marginal on the input is not the input but its transpose; see Refs.~\cite{fullwood2022quantum, parzygnat2023time} for details), it is nonetheless very close in spirit to a state over time. In the fully commuting case, it reduces to the conventional joint input-output probability distribution $P_{\rm fwd}(x,y)$. For this reason, we will borrow the notation from Ref.~\cite{leifer2013towards} and define $\map{E}\star\rho$ as}

\begin{align}
    \map{E}\star\rho := \left(\openone_B \otimes \sqrt{\rho^T}\right)C_{\map{E}}\left(\openone_B \otimes \sqrt{\rho^T}\right),
\end{align}
so that our forward process operator becomes
\begin{align} \label{eq:Q_F}
    Q_{\rm fwd}= \map{E}\star\gamma\;.
\end{align}

}

For the reverse process, the representation is defined similarly, but with respect to a reference state $\tau\in S(\HS_B)$ on the output, and with an added transpose, in formula
\begin{align} \label{eq:Q_R}
    Q_{\rm rev}:= (\map{R}\star\tau)^T = (\sqrt{\tau}\otimes \openone_A)\;C_{\map{R}}^T\;(\sqrt{\tau}\otimes \openone_A)\;,
\end{align}
Note that the ordering of systems $A$ and $B$ are swapped so that $C_{\map{R}}=\sum_{k,l}\ketbra{k}{l}\otimes\map{R}(\ketbra{k}{l})\in\mL_+(\HS_B\otimes\HS_A)$ and $Q_{\rm rev}\in\mS(\HS_B\otimes\HS_A)$, thus matching \cref{eq:Q_F}. The same representation for the reverse processes has been used for a definition of observational entropy with general priors~\cite{bai2023observational}. In this way, the marginal states become
\begin{align}
    \Tr_A[Q_{\rm rev}] = \tau\;,\qquad    \Tr_B[Q_{\rm rev}] = \map{R}(\tau)^T\;.
\end{align}

\medskip{\em The quantum minimum change principle.}---We are now ready to formulate the quantum minimum change principle in a rigorous way. Here, we choose to measure ``change'' in terms of the (square-root) \textit{fidelity}, defined as  %
\cite{jozsa1994fidelity,uhlmann1976transition}
\begin{align} \label{eq:fidelity}
    F(\rho,\sigma):= \Tr\left[\sqrt{\sqrt{\rho}\,\sigma\sqrt{\rho}}\right] = \Tr\left[\sqrt{\sqrt{\sigma}\rho\sqrt{\sigma}}\right].
\end{align}
Fidelity is one of the most natural measures of the closeness between quantum states and has found countless applications in quantum information theory.
Here, we use the fidelity to measure the statistical similarity between the forward process, represented by the bipartite state $Q_{\rm fwd}$, and the reverse process, represented by $Q_{\rm rev}$.
Since both $Q_{\rm fwd}$ and $Q_{\rm rev}$ are well-defined density matrices, the fidelity between them $F(Q_{\rm fwd},Q_{\rm rev})$ is also well-defined, regardless of how $\mE$, $\gamma$, $\mR$, and $\tau$ are chosen. In particular, it is always bounded between 0 and 1, the latter being achieved if and only if $Q_{\rm fwd}=Q_{\rm rev}$.

The quantum minimum change principle, similar to \cref{eq:clax_opt_gamma_tau}, minimizes the deviation between $Q_{\rm fwd}=\map{E} \star \gamma$ and $Q_{\rm rev}=(\map{R}\star\tau )^T$, or equivalently, maximizes their fidelity. The principle is formally expressed as follows:
\begin{align} \label{eq:main_program}
\begin{split}
    \max_{\map{R}}\quad&  F\left( \map{E} \star \gamma , (\map{R}\star\tau )^T \right) \\
    \text{subject to}\quad& \map{R} \text{\rm~is CPTP}\;.
\end{split}
\end{align}
Our central result is to show that, given that $\map{E} \star \gamma$ and $\tau$ are both strictly positive (thus, so must be $\gamma$ and $\map{E}(\gamma)$ as well), the above program has a unique solution, for which we provide a closed-form expression. Remarkably, we find that whenever $\mE(\gamma)$ and $\tau$ commute, our solution coincides with the \textit{Petz transpose map}~\cite{petz1,petz} computed for the forward channel $\mE$ with respect to the prior $\gamma$, independent of $\tau$. This independence means that the Petz transpose map is uniformly optimal over all choices of $\tau$, similar to the behavior of the classical Bayes' rule mentioned earlier.

\begin{theorem} \label{thm:petz_optimal}
Given a forward CPTP map $\map{E}$, a reference input $\gamma$ and a reference output $\tau$, assuming both $\map{E}\star\gamma>0$ and $\tau>0$, the following CPTP map
    \begin{align} \label{eq:Petz_var}
        &\map{R}(\sigma) := \sqrt{\gamma}\;\map{E}^\dag(D \sigma D^\dag )\sqrt{\gamma} \\
        &D:=\sqrt{\tau}\left(\sqrt{\tau}\map{E}(\gamma)\sqrt{\tau}\right)^{-1/2}\label{eq:operator-D}
    \end{align}
    is the unique solution of the program \cref{eq:main_program}. Furthermore, if $[\tau,\map{E}(\gamma)]=0$, the above solution coincides with the Petz transpose map of $\mE$ with respect to $\gamma$, i.e.,
    \begin{align} \label{eq:Petz}
        \map{R}(\sigma) = \sqrt{\gamma}\;\map{E}^\dag\left(\map{E}(\gamma)^{-1/2}\sigma\map{E}(\gamma)^{-1/2}\right)\sqrt{\gamma} \,.
    \end{align}
\end{theorem}

{An example comparing the Petz transpose map and the optimal solution of \cref{eq:main_program} is shown in \cref{fig:curve}.}

\begin{figure}
    \centering
    \includegraphics[width=0.9\linewidth]{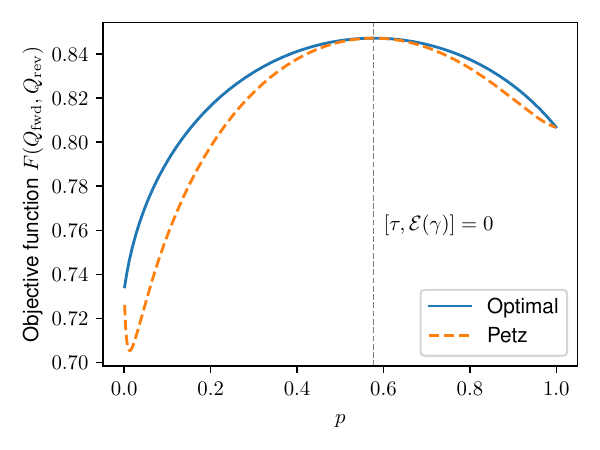}
    \caption{Example comparing the Petz transpose map with the optimal map obtained from the minimum change principle. The forward map $\mathcal{E}$ models one step of a thermalizing quantum machine where the system interacts with a thermal state $\xi$ \cite{scarani2002}. Concretely, we consider qubit systems, and the map is defined as $\mathcal{E}(\rho):=\Tr_{A'}[V(\rho\otimes\xi_{A'})V^\dag]$, where $V:=\cos\frac{\pi}{8}~ \openone + i\sin\frac{\pi}{8}~ U_{\rm sw}$ is the partial swap operator with angle $\frac{\pi}{8}$, $U_{\rm sw}$ is the SWAP gate, and $\xi=0.95\ketbra{0}{0} + 0.05\ketbra{1}{1}$. The prior belief is parameterized as $\gamma = (1-p)\ketbra{0}{0}+p\ketbra{1}{1}$, and the reference output state $\tau$ is chosen as $\tau=H\xi H$, where $H$ is the Hadamard gate. The figure plots the value of the objective function of \cref{eq:main_program} over $p\in[0.001,0.999]$, for $\mathcal{R}$ chosen to be the Petz transpose map [\cref{eq:Petz}, orange dashed line] or the optimal map [\cref{eq:Petz_var}, blue solid line]. In most cases $[\tau,\mathcal{E}(\gamma)]\neq 0$, and the Petz transpose map is suboptimal; while at $p = 1.85 - 0.9\sqrt{2} \approx 0.577$, where $\mathcal{E}(\gamma) = \openone/2$ and $[\tau, \mathcal{E}(\gamma)] = 0$, the Petz map achieves optimality.}
    \label{fig:curve}
\end{figure}

Before fleshing out the proof of the above theorem (the full details can be found in Supplemental Material \footnote{\label{supplemental}See Supplemental Material for proofs the main theorem and lemmas, which includes Refs. \cite{kubo1980means, nakamura2007geometric,stewart1990matrix,bhatia2018strong,del2018taylor}}), let us first make some comments. First, it is easy to verify that the map given in \cref{eq:Petz_var} is indeed completely positive and trace-preserving. Complete positivity holds because $\mR$ is the composition of three completely positive maps, that is, $D\bullet D^\dag$, $\mE^\dag(\bullet)$, and $\sqrt{\gamma}\bullet\sqrt{\gamma}$. The trace-preservation condition is $\Tr[\map{R}(\sigma)]=\Tr[D^\dag \map{E}(\gamma)D\sigma]=\Tr[\sigma]$ for any $\sigma$, equivalent to the condition
\[
\left(\sqrt{\tau}\map{E}(\gamma)\sqrt{\tau}\right)^{-1/2}\sqrt{\tau}\map{E}(\gamma)\sqrt{\tau}\left(\sqrt{\tau}\map{E}(\gamma)\sqrt{\tau}\right)^{-1/2}=\openone\;,
\]
which holds because it is equivalent to $A^{-1/2}AA^{-1/2}=\openone$, for $A>0$.

Second, similar to \cref{eq:clax_opt_gamma_tau}, the optimization contains $\tau$ as a parameter, but the solution \cref{eq:Petz} is uniformly optimal regardless of the choice of $\tau$, as long as $[\tau,\map{E}(\gamma)]=0$. The condition $[\tau,\map{E}(\gamma)]=0$, which also guarantees that our solution coincides with the Petz transpose map, is satisfied in some situations of physical interest, such as
\begin{enumerate}
    \item when $\tau=\map{E}(\gamma)$, as is often assumed when considering quantum error correction~\cite{barnum2002reversing} or quantum fluctuation relations~\cite{crooks-reversal}; notice that this case leads to the optimal solution with $F(Q_{\rm fwd},Q_{\rm rev})=1$, namely $Q_{\rm fwd}=Q_{\rm ref}$, indicating that, if the new information fully agrees with what was expected given the prior, we recover not only the initial state, but also the entire process, making the forward and reverse processes identical, as it happens in the classical case;
    \item when $\tau=\openone_B/d_B$, corresponding to ``uninformative'' new information;
    \item when $\mE(\gamma)=\openone_B/d_B$, corresponding to the case where the channel is maximally uninformative for the chosen prior; \item when $\mathcal{E}$ is a quantum-classical channel, e.g., a measurement channel, and the new information represented by $\tau$ is classical information about the measurement results.
\end{enumerate}
The last condition makes Theorem~\ref{thm:petz_optimal} especially compelling in the context of observational entropy and the corresponding second law~\cite{buscemi2022observational,bai2023observational,nagasawa2024observational}.

{Third, when the commutativity condition $[\tau, \map{E}(\gamma)] = 0$ is satisfied and the retrodictive channel does not depend on $\tau$, our solution is in line with Jeffrey’s framework of \textit{probability kinematics}~\cite{jeffrey}, where the update rule relies solely on the prior $\gamma$ and the forward map $\varphi$, independently of the new information $\tau$. However, when $\mE(\gamma)$ and $\tau$ do not commute, the quantum case departs from Jeffrey's framework, and the role of $\tau$ becomes interpretationally significant. If $\tau$ is taken as a subjective prior belief about what the observer expected to observe before obtaining the actual outcome, then its sole purpose is to ``guide'' the optimization according to the principle of minimum change. In this case, the update rule is given by \cref{eq:Petz_var} with $\sigma$ representing the newly acquired information, and the rule remains independent of $\tau$. In contrast, if $\tau$ is identified with the \textit{actual} observed evidence---namely, if one sets $\sigma = \tau$ in \cref{eq:Petz_var}---then the quantum counterpart of Jeffrey’s retrodicted distribution $\gamma'(x)$ from \cref{eq:jeffrey} becomes the quantum state}
\begin{align*}
&\gamma':=\map{R}(\tau)=\sqrt{\gamma}\map{E}^\dag(D \tau D^\dag)\sqrt{\gamma}\;.
\end{align*}
While the CPTP map $\map{R}$ is always linear, the \textit{overall} dependence of $\gamma'$ on $\tau$ is not, in contrast to the classical case.
Again, it is the commutativity condition $[\tau,\map{E}(\gamma)]=0$ that restores the linear dependence of $\gamma'$ on $\tau${, and thus in the classical case the two interpretations make no difference}. The functional dependence of the retrodicted state on the new information is an important aspect of the problem, which we leave open for future research, especially in view of the axiomatic approach proposed in~\cite{parzygnat2023axioms}, where linearity of the retrodiction is assumed from the start.

Fourth, when $[\tau,\map{E}(\gamma)]=0$, Theorem~\ref{thm:petz_optimal} can also deal with a non-invertible $\tau$: in this case, we can replace $\tau$ with $\tau_\epsilon:=(1-\epsilon)\tau+\epsilon u$, where $\epsilon>0$ and $u$ denotes the uniform (i.e.~maximally mixed) state; find the optimal reverse channel, which does not depend on $\tau_\epsilon$ as the latter preserves the commutation relation; and finally take the limit for $\epsilon\to 0$. In this way, we see that Theorem~\ref{thm:petz_optimal} does indeed contain the conventional Bayes' rule~\eqref{eq:bayes}, where the new information is given as a delta function.

Our proof is based on the method of Lagrangian multipliers. Here, we give only an outline and refer the interested reader to \cref{appendix} for the complete argument. First, write the program in terms of the Choi operator $C_{\map{R}}$ of $\map{R}$. Recalling that the reverse process goes from $B$ back to $A$, the trace-preserving constraint becomes $\Tr_A[C_{\map{R}}] - \openone_B=0$. 
Then, define the Lagrangian of the problem as
\begin{align}
    \mathcal{L}:=F\left( Q_{\rm fwd}, Q_{\rm rev} \right) + \Tr[\Lambda(\Tr_A[C_\mathcal{R}] - \openone_B)] \;,
\end{align}
where $\Lambda$ is the Lagrangian multiplier to enforce the trace-preserving constraint. Note that the complete positivity constraint is not explicitly invoked. The solution \cref{eq:Petz_var} corresponds to the Choi operator 
\begin{align} \label{eq:C_R}
    C_{\map{R}}=\left(D^T \otimes \sqrt{\gamma}\right) C_{\map{E}}^T \left(D^*\otimes \sqrt{\gamma}\right) \,,
\end{align}
where $D^*$ is the complex conjugate of the operator $D$ defined in \cref{eq:operator-D},  computed {with respect to the same basis as the transpose}. %
In \cref{appendix}, we show that \cref{eq:C_R} and the value of $\Lambda$ given by
\begin{align}
    \Lambda = -\frac12\left(\sqrt{\sqrt{\tau}\map{E}(\gamma)\sqrt{\tau}}\right)^T
\end{align}
satisfy the condition that the partial derivative of $\mathcal{L}$ is zero in every direction. Therefore, $C_{\map{R}}$ is a local optimum of the program. 
Finally, due to the strict concavity of fidelity under the theorem's hypotheses (proven as a separate lemma in \cref{appendix}), we conclude that $C_{\map{R}}$ is the unique global maximum.

\medskip{\em Discussion.---}In this work, we have generalized the minimum change principle to the quantum case, offering a new formulation of the quantum Bayes' rule. In particular, when fidelity is used as a figure of merit, the Petz transpose map, which is often only \textit{pretty good} but suboptimal \cite{hausladen1994pretty,zheng2024near}, has naturally emerged as the \textit{unique optimal} solution in many relevant cases, confirming the central role of the Petz transpose map as a quantum Bayes' rule. The consistency between the theory of statistical sufficiency (in which the Petz transpose map plays a central role) and the variational principle of minimum change suggests a wide range of applicability for that principle in all areas where the Petz transpose map has appeared to play a role, such as quantum information theory, quantum statistical mechanics, and many-body physics.

{
Besides fidelity, one may consider other divergences as the figure of merit, such as the Hilbert--Schmidt distance \cite{buvzek1996quantum} or the Umegaki \cite{umegaki,umegaki1962conditional} or Belavkin--Staszewski \cite{belavkin1982c,matsumoto2015new} quantum relative entropies, and wonder whether they lead to the same quantum Bayes' rule. In the special case of $\tau=\map{E}(\gamma)$, they always do, because the Petz transpose map produces $Q_{\rm fwd}=Q_{\rm ref}$ and any reasonable divergence should reach optimum for two identical operators. However, in the general case $\tau\neq\map{E}(\gamma)$, even if $[\tau,\map{E}(\gamma)]=0$, numerical optimization shows that the optimum points may differ for different divergences. It is thus interesting to explore whether different quantum Bayes' rules will arise from those various  choices.
}

Further, by imposing additional constraints to the program \cref{eq:main_program}, one could restrict the reverse process to a desired subset. We have solved the program \cref{eq:main_program} analytically, while for general cases, convex additional constraints maintains the convexity of the optimization program, for which efficient numerical algorithms can be adopted \cite{boyd2004convex,brahmachari2023fixed}.
With this approach, the minimum change principle may be extended to quantum combs \cite{chiribella2008quantum,chiribella2009theoretical}, quantum supermaps \cite{chiribella2008transforming,bai2024bayesian}, and quantum Bayesian networks \cite{tucci1995quantum,warmuth2005bayes,tucci2007factorization,leifer2008quantum,barrett2019quantum}, offering new belief update rules for them. The tools introduced in this work may also pave the way towards a fully quantum generalisation of entropy production and fluctuation theorems \cite{BS21,AwBS}.

\medskip\textit{Acknowledgments.}---The authors would like to thank Maarten Wegewijs for helpful comments. G.\ B.\ and V.\ S.\ acknowledge support from the National Research Foundation, Singapore and A*STAR under its CQT Bridging Grant; from the National Research Foundation, Singapore, through the National Quantum Office, hosted in A*STAR, under its Centre for Quantum Technologies Funding Initiative (S24Q2d0009); and from the Ministry of Education, Singapore, under the Tier 2 grant ``Bayesian approach to irreversibility'' (Grant No.~MOE-T2EP50123-0002).
F.\ B.\ acknowledges support from MEXT Quantum Leap Flagship Program (MEXT QLEAP) Grant No.~ JPMXS0120319794, from MEXT-JSPS  Grant-in-Aid for Transformative Research Areas  (A) ``Extreme Universe,'' No.~21H05183, and  from JSPS  KAKENHI Grant No.~23K03230.
G.\ B.\ acknowledges support from the Start-up Fund (Grant No.~G0101000274) from The Hong Kong University of Science and Technology (Guangzhou).

\bibliography{notes_modular}

\newpage
\onecolumngrid
\appendix
\newpage

\section{Proof of Theorem \ref{thm:petz_optimal}\label{appendix}}

We first give some lemmas useful for the proof.

\begin{lemma} \label{lem:XBX=A}
    For positive operator $A$ and full-rank positive operator $B$, the equation of $XBX=A$ has a unique positive solution $X= B^{-1/2}\sqrt{\sqrt{B}A\sqrt{B}}B^{-1/2} $.
\end{lemma}
The equation a special case of the algebraic matrix Riccati equation, and the solution has been shown in \cite{kubo1980means, nakamura2007geometric}. We present a proof here for completeness.

\begin{proof}
    
\begin{align}
    XBX&=A \\
    \sqrt{B}X\sqrt{B}\sqrt{B}X\sqrt{B}&=\sqrt{B}A\sqrt{B} \\
    \sqrt{B}X\sqrt{B} &= \sqrt{\sqrt{B}A\sqrt{B}} \\
    X&= B^{-1/2}\sqrt{\sqrt{B}A\sqrt{B}}B^{-1/2} \label{eq:solutionXBX_A}
\end{align}
For positive operators $A$, $X$ and full-rank positive operator $B$, every step above is sufficient and necessary. Therefore, the solution of $X$ exists and is unique.
\end{proof}

\bigskip A corollary of this lemma is that, for full-rank positive operators $A$ and $B$, and for $X$ in \cref{lem:XBX=A}, $X^{-1}$ is a positive solution for the equation $X^{-1}AX^{-1}=B$. Then by \cref{lem:XBX=A}, $X^{-1}=A^{-1/2}\sqrt{\sqrt{A}B\sqrt{A}}A^{-1/2}$. Combining this with \cref{eq:solutionXBX_A}, one gets
\begin{align} \label{eq:BBABB=AABAA}
    B^{-1/2}\sqrt{\sqrt{B}A\sqrt{B}}B^{-1/2} = \sqrt{A}\left(\sqrt{A}B\sqrt{A}\right)^{-1/2}\sqrt{A}
\end{align}
for full-rank positive operators $A$ and $B$.

The lemma below computes the directional derivative, or Fr\'{e}chet derivative of the quantum fidelity. %

\begin{lemma} \label{lem:diff_fidelity}
    For full-rank positive operators $\rho$ and $\sigma$, the total differential of their fidelity  $F(\rho,\sigma):= \Tr\left[\sqrt{\sqrt{\sigma}\rho\sqrt{\sigma}}\right]$, with respect to directions $\partial\rho$ and $\partial\sigma$, is
    \begin{align}
        \partial F(\rho,\sigma) &= \frac12\Tr[\Delta \partial \sigma  ] + \frac12\Tr[\Delta^{-1} \partial \rho  ] 
    \end{align}
    where $\Delta:=\sigma^{-1/2}\sqrt{\sqrt{\sigma}\rho\sqrt{\sigma}}\sigma^{-1/2}$.
\end{lemma}
\begin{proof}

The operator $\Delta$ satisfies 
\begin{align}
    \Delta \sigma \Delta = \rho , ~\Delta^{-1}\rho\Delta^{-1} = \sigma \,.
\end{align}
Then one can rewrite the expression for fidelity as
\begin{align}
    F(\rho,\sigma) &= \Tr\left[\sqrt{\sqrt{\sigma}\rho\sqrt{\sigma}}\right] = \Tr[\sigma \Delta ] = \Tr[\rho \Delta^{-1}] \\
    & = \frac12\left(\Tr[\sigma \Delta ] + \Tr[\rho\Delta^{-1}]\right)\,.
\end{align}

The total differential of the fidelity is
\begin{align}
    \partial F(\rho,\sigma) 
    &= \frac12\left(\partial \Tr[\sigma \Delta ] + \partial\Tr[\rho\Delta^{-1}]\right) \\
    & =  \frac12\left( \Tr[\Delta \partial \sigma] + \Tr[\sigma \partial \Delta] + \Tr[\rho \partial\Delta^{-1}] + \Tr[\Delta^{-1} \partial\rho]\right) \\
    & =  \frac12\left( \Tr[\Delta \partial \sigma] + \Tr[\sigma \partial \Delta] - \Tr[\rho \Delta^{-1}\partial\Delta \cdot\Delta^{-1}] + \Tr[\Delta^{-1} \partial\rho]\right) \label{eq:use_diff_inv} \\
    & =  \frac12\left( \Tr[\Delta \partial \sigma] + \Tr[(\sigma -\Delta^{-1} \rho \Delta^{-1})\partial \Delta ] + \Tr[\Delta^{-1} \partial\rho]\right) \\
    & =  \frac12\Tr[\Delta \partial \sigma  ]  + \frac12\Tr[\Delta^{-1} \partial\rho]
\end{align}

In \cref{eq:use_diff_inv}, we have used the formula for the differential of the inverse of an matrix \cite{stewart1990matrix}:
\begin{align} \label{eq:diffinv}
    \partial\Delta^{-1} = -\Delta^{-1}\partial\Delta\cdot\Delta^{-1}
\end{align}
\end{proof}

\bigskip
\begin{proof}[Proof of \cref{thm:petz_optimal}.]
We first prove the local optimality of the solution given by \cref{eq:Petz_var} by deriving it from the Lagrangian multiplier method, and then use the strict concavity of fidelity (\cref{lem:fidelity_strictly_concave}) to show the global optimality and uniqueness.

We use the definition of $Q_{\rm fwd}$ and $Q_{\rm rev}$ in \cref{eq:Q_F,eq:Q_R}. 
Since we are first proving the existence of a local optimal solution of \cref{eq:main_program} in the form \cref{eq:Petz_var}, it suffices to restrict $C_{\map{R}}$ to be full-rank, and then $Q_{\rm rev}$ is full-rank. 
Recall that the variable we optimize over is $\map{R}$, or equivalently the transpose of its Choi operator $C_{\map{R}}^T$.

Define the Lagrangian as
\begin{align}
    \mathcal{L}:=F\left( Q_{\rm fwd}, Q_{\rm rev} \right) + \Tr[\Lambda(\Tr_A[C_\mathcal{R}] - \openone_B)]
\end{align}
where $\Lambda$ is the Lagrangian multiplier for the trace-preserving condition $\Tr_A[C_{\map{R}}] - \openone_B=0$.
The complete positivity constraint is not explicitly invoked. Defining
\begin{align}
    \Delta:=Q_{\rm rev}^{-1/2}\sqrt{\sqrt{Q_{\rm rev}}Q_{\rm fwd}\sqrt{Q_{\rm rev}}}Q_{\rm rev}^{-1/2}\,,
\end{align}
then $Q_{\rm rev} = \Delta^{-1}Q_{\rm fwd}\Delta^{-1}$ and by \cref{eq:Q_R}, $C_{\mathcal{R}}$ satisfies
\begin{align} \label{eq:C_R_as_Q_F}
    C_{\mathcal{R}}^T = (\tau^{-1/2} \otimes \openone_A)Q_{\rm rev}(\tau^{-1/2} \otimes \openone_A) = (\tau^{-1/2} \otimes \openone_A)\Delta^{-1}Q_{\rm fwd}\Delta^{-1}(\tau^{-1/2} \otimes \openone_A) \,.
\end{align}
For $Q_{\rm fwd}$ being a constant, by \cref{lem:diff_fidelity}, $\partial F(Q_{\rm fwd},Q_{\rm rev}) = \frac12\Tr[\Delta \partial Q_{\rm rev}  ] $. The total differential of $\mathcal{L}$ is
\begin{align}
    \partial\mathcal{L} &= \frac12\Tr\left[\Delta \partial Q_{\rm rev}  \right] +\Tr\left[(\Lambda \otimes \openone_A)\partial C_\mathcal{R}\right]+ \Tr\left[(\Tr_A[C_\mathcal{R}\right] - \openone_B)\partial\Lambda] \\
    & = \frac12\Tr\left[(\sqrt\tau \otimes \openone_A) \Delta (\sqrt\tau \otimes \openone_A) \partial C_{\mathcal{R}}^T  \right] +\Tr\left[(\Lambda^T \otimes \openone_A) \partial C_\mathcal{R}^T\right]+ \Tr\left[(\Tr_A[C_\mathcal{R}] - \openone_B)\partial\Lambda\right]\;.
\end{align}
In the second equality we also used the fact that $\tau$ is considered a constant. 

The local optimum is attained when the differential of $\mathcal{L}$ is zero with respect to both $C_{\mathcal{R}}^T$ and $\Lambda$. This gives
\begin{align}
    \frac12(\sqrt\tau \otimes \openone_A) \Delta (\sqrt\tau \otimes \openone_A) + \Lambda^T \otimes \openone_A &= 0 \label{eq:L=0_1}\\
    \Tr_A[C_\mathcal{R}] - \openone_B &= 0  \label{eq:L=0_2}
\end{align}
From \cref{eq:L=0_1}, one gets
\begin{align} \label{eq:Delta_as_D}
    \Delta = -2\tau^{-1/2} \Lambda^T \tau^{-1/2} \otimes \openone_A =: \Delta_B\otimes \openone_A
\end{align}
where $\Delta_B:=-2\tau^{-1/2} \Lambda^T \tau^{-1/2}$. $\Delta_B>0$ since $\Delta>0$. By \cref{eq:L=0_2}, one gets $\openone_B = \Tr_A[C_\mathcal{R}^T]$. Substituting \cref{eq:C_R_as_Q_F,eq:Delta_as_D} into this, one gets
\begin{align}
    \openone_B &= \Tr_A[C_\mathcal{R}^T] \\
    & = \Tr_A\left[(\tau^{-1/2} \otimes \openone_A)\Delta^{-1}Q_{\rm fwd}\Delta^{-1}(\tau^{-1/2} \otimes \openone_A)\right]\\
    & = \Tr_A\left[(\tau^{-1/2} \Delta_B^{-1} \otimes \openone_A)Q_{\rm fwd}(\Delta_B^{-1}\tau^{-1/2} \otimes \openone_A)\right]\\
    & = \tau^{-1/2} \Delta_B^{-1} \ \Tr_A\left[Q_{\rm fwd}\right]\ \Delta_B^{-1}\tau^{-1/2}\\
    & = \tau^{-1/2} \Delta_B^{-1}\ \map{E}(\gamma)\ \Delta_B^{-1}\tau^{-1/2}    \,,
\end{align}
where the last equation comes from \cref{eq:Q_F_marginals}. This can be equivalently written as a equation for $\Delta_B$, namely %
\begin{align}
    \Delta_B \tau \Delta_B =  \map{E}(\gamma) \,,
\end{align}
solving which with \cref{lem:XBX=A} obtains
\begin{align} \label{eq:Delta_B}
    \Delta_B = \tau^{-1/2}\sqrt{\sqrt{\tau}\map{E}(\gamma)\sqrt{\tau}}\tau^{-1/2}\;. 
\end{align}
In turn, this gives the value of the Lagrangian multiplier by
\begin{align}
    \Lambda^T = -\frac12\sqrt{\tau}\Delta_B\sqrt{\tau}= -\frac12\sqrt{\sqrt{\tau}\map{E}(\gamma)\sqrt{\tau}}.
\end{align}
Substituting \cref{eq:Delta_B} into  \cref{eq:Delta_as_D}, one gets
\begin{align}
    \Delta^{-1} &= %
    \Delta_B^{-1}\otimes\openone_A = 
    \sqrt{\tau}\left(\sqrt{\tau}\map{E}(\gamma)\sqrt{\tau}\right)^{-1/2}\sqrt{\tau} \otimes \openone_A \,.
\end{align}
Then, by \cref{eq:C_R_as_Q_F},
\begin{align}
    \mathrel{\phantom{=}}C_{\mathcal{R}}^T &= (\tau^{-1/2} \otimes \openone_A)\Delta^{-1}Q_{\rm fwd}\Delta^{-1}(\tau^{-1/2} \otimes \openone_A)\\
    & = \left(\left(\sqrt{\tau}\map{E}(\gamma)\sqrt{\tau}\right)^{-1/2}\sqrt{\tau} \otimes \openone_A\right) Q_{\rm fwd} \left(\sqrt{\tau}\left(\sqrt{\tau}\map{E}(\gamma)\sqrt{\tau}\right)^{-1/2}\otimes \openone_A\right)\\
    & = \left(\left(\sqrt{\tau}\map{E}(\gamma)\sqrt{\tau}\right)^{-1/2}\sqrt{\tau} \otimes \sqrt{\gamma^T}\right) C_{\map{E}} \left(\sqrt{\tau}\left(\sqrt{\tau}\map{E}(\gamma)\sqrt{\tau}\right)^{-1/2}\otimes \sqrt{\gamma^T}\right) \\
    & = \left(D^\dag \otimes \sqrt{\gamma^T}\right) C_{\map{E}} \left(D\otimes \sqrt{\gamma^T}\right) . \label{eq:CRT_as_CE}
\end{align}
where $D := \sqrt{\tau}\left(\sqrt{\tau}\map{E}(\gamma)\sqrt{\tau}\right)^{-1/2}$. The above Choi operator corresponds to the map in \cref{eq:Petz_var}. To see this, assuming the Kraus decomposition of $\map{E}$ is $\map{E}(\rho)=\sum_k K_k \rho K_k^\dag$, one computes $C_{\map{R}}^T$ from \cref{eq:Petz_var} by definition of the Choi operator $C_{\map{R}}=\sum_{ij}\ketbra{i}{j} \otimes \map{R}(\ketbra{i}{j})$ of the map $\map{R}(\sigma):=\sqrt{\gamma}\,\map{E}^\dag(D\sigma D^\dag)\sqrt{\gamma}$ as
\begin{align}
    \mathrel{\phantom{=}}C_{\mathcal{R}}^T &= \left(\sum_{ij}\ketbra{i}{j} \otimes \sqrt{\gamma}\,\map{E}^\dag \left(D\ketbra{i}{j} D^\dag\right) \sqrt{\gamma}\right)^T\\
    & = \sum_{ijk}\left(\ketbra{i}{j} \otimes \sqrt{\gamma}\, K_k^\dag D\ketbra{i}{j} D^\dag K_k \sqrt{\gamma}\right)^T\\
    & = \sum_{k}\left((\openone_B\otimes \sqrt{\gamma}) \left(\sum_{ij}\ketbra{i}{j} \otimes \ K_k^\dag D\ketbra{i}{j} D^\dag K_k \right) (\openone_B\otimes \sqrt{\gamma})\right)^T \,.
\end{align}
Using the identity $\sum_{ij}\ketbra{i}{j} \otimes X \ketbra{i}{j} Y = \sum_{ij} X^T\ketbra{i}{j} Y^T \otimes \ketbra{i}{j}$ for operators $X$ and $Y$  on $\HS_B$, one has 
\begin{align}
    C_{\mathcal{R}}^T& =  \sum_{k}\left((\openone_B\otimes \sqrt{\gamma}) \left(\sum_{ij} (K_k^\dag D)^T\ketbra{i}{j} (D^\dag K_k)^T \otimes \ketbra{i}{j} \right) (\openone_B\otimes \sqrt{\gamma})\right)^T\\ 
    & = \sum_{ijk}\left((K_k^\dag D)^T\ketbra{i}{j} (D^\dag K_k)^T \otimes \sqrt{\gamma} \ketbra{i}{j} \sqrt{\gamma}\right)^T \\
    & = \sum_{ijk}D^\dag K_k \ketbra{j}{i} K_k^\dag D \otimes \sqrt{\gamma^T} \ketbra{j}{i} \sqrt{\gamma^T} \\
    & = \sum_{ij} D^\dag \map{E}(\ketbra{j}{i}) D \otimes \sqrt{\gamma^T} \ketbra{j}{i} \sqrt{\gamma^T} \\
    & = \left(D^\dag \otimes \sqrt{\gamma^T}\right) C_{\map{E}} \left(D\otimes \sqrt{\gamma^T}\right),
\end{align}
which is the same as \cref{eq:CRT_as_CE}.

Therefore, $\map{R}$ in \cref{eq:Petz_var} satisfies $\partial\mathcal{L}=0$ and is a local optimum of \cref{eq:main_program}. %

Among full-rank positive operators $C_\mathcal{R}$, this solution is the unique global maximum since, as shown in \cref{lem:fidelity_strictly_concave} below, the objective function $F(Q_{\rm fwd},Q_{\rm rev})$ is strictly concave with respect to $C_\mathcal{R}$. More precisely, \cref{lem:fidelity_strictly_concave} below shows that $F(Q_{\rm fwd},Q_{\rm rev})$ is strictly concave when $Q_{\rm fwd}>0$ and $Q_{\rm rev}>0$, and $Q_{\rm rev}$ is a linear function of $C_\mathcal{R}$. Since the solution \cref{eq:CRT_as_CE} has $C_\mathcal{R}>0$, due to the continuity of quantum fidelity, this solution is also the unique global maximum among all positive operators including non-full-rank ones. That is because the set of positive operators is the closure of full-rank positive operators.

Last, if $[\tau,\map{E}(\gamma)] = 0$, one has $D=\map{E}(\gamma)^{-1/2}$, and the above equation simplifies to
\begin{align}
    C_{\mathcal{R}}^T %
    & = \left(\map{E}(\gamma)^{-1/2} \otimes \sqrt{\gamma^T}\right) C_{\map{E}} \left(\map{E}(\gamma)^{-1/2}\otimes \sqrt{\gamma^T}\right) . \label{eq:res_Petz_choi}
\end{align}
\cref{eq:res_Petz_choi} exactly corresponds to the Petz map \cref{eq:Petz}.
\end{proof}

\bigskip The next lemma shows the strict concavity of fidelity between two full-rank operators. We used it to show the uniqueness of the optimal solution in \cref{thm:petz_optimal}. A stronger version of this lemma is shown in~\cite{bhatia2018strong} but we present here an independent proof for the sake of completeness.

\begin{lemma}\label{lem:fidelity_strictly_concave}
    When $\rho>0$ and $\sigma>0$, the fidelity $F(\rho,\sigma)$ is strictly concave with respect to $\sigma$.
\end{lemma}
\begin{proof}
    Considering $F(\rho,\sigma)$ as a multivariate function of elements of $\sigma$, its strict concavity is equivalent to its second-order derivative being strictly negative in every direction.

    The domain of fidelity $F$ consists of positive Hermitian operators, and thus a perturbation of $\sigma$ should also be Hermitian. Let us consider a perturbation of $\sigma$ as $\sigma+\varepsilon H$, where $\varepsilon\in\mathbb{R}$ and $H\neq0$ is a Hermitian operator.

    Let $f(\varepsilon):=F(\rho,\sigma+\varepsilon H)$.
    The goal is to show that the second-order derivative of $f$ at $\varepsilon=0$ is negative in every direction $H\neq 0$.

    We start with the gradient of $F(\rho,\sigma)$ given by \cref{lem:diff_fidelity}:
    \begin{align}
        \partial F(\rho,\sigma) = \frac12 \Tr[\Delta \partial\sigma],%
    \end{align}
    where $\Delta=\sigma^{-1/2}\sqrt{\sqrt{\sigma}\rho\sqrt{\sigma}}\,\sigma^{-1/2}$. This gives the first-order derivative 
    \begin{align}\label{eq:df1}
        \left.\frac{df}{d\varepsilon}\right|_{\varepsilon=0}= \frac12 \Tr[\Delta H]\,.
    \end{align}

    Now, we compute the second-order derivative of $f$, which involves the differential of $\Delta$. Let $M:=\sqrt{\sqrt{\rho}\,\sigma\sqrt{\rho}}$, its differential is
    \begin{align} \label{eq:diffM}
        \partial M & = \partial\sqrt{\sqrt{\rho}\,\sigma\sqrt{\rho}} = \int_0^\infty e^{-tM} \sqrt{\rho}~\partial\sigma\,\sqrt{\rho} e^{-tM} dt \,,
    \end{align}
    where we have used \cite{del2018taylor}
    \begin{align}
        \partial \sqrt{A} = \int_0^\infty e^{-t\sqrt{A}} ~ \partial A ~ e^{-t\sqrt{A}} dt \,.
    \end{align}
    Since $\rho>0$ and $\sigma>0$, one has $M>0$.

    By \cref{eq:BBABB=AABAA}, $\Delta = \sqrt{\rho}\left(\sqrt{\rho}\,\sigma\sqrt{\rho}\right)^{-1/2}\sqrt{\rho} = \sqrt{\rho}M^{-1}\sqrt{\rho}$. Using \cref{eq:diffinv,eq:diffM}, its differential is
    \begin{align}
        \partial\Delta &= \sqrt{\rho} (\partial M^{-1}) \sqrt{\rho} \\
        &= -\sqrt{\rho}\,M^{-1}(\partial M)M^{-1}\sqrt{\rho}\\
        & = -\sqrt{\rho}\,M^{-1}\int_0^\infty e^{-tM} \sqrt{\rho}~\partial\sigma\,\sqrt{\rho}\, e^{-tM} dt \,M^{-1}\sqrt{\rho}\\
        & = -\int_0^\infty K(t)\, \partial\sigma\, K(t) dt\,,
    \end{align}
    where $K(t):=\sqrt{\rho}\,M^{-1}e^{-tM} \sqrt{\rho} = \sqrt{\rho}\,e^{-tM}M^{-1} \sqrt{\rho}$. $M>0$ implies $M^{-1}e^{-tM}>0$, and since $\rho>0$, one has $K(t)>0$ for any $t\geq0$. We then write the derivative of $\Delta$ with respect to $\varepsilon$ as
    \begin{align}
        \left.\frac{d\Delta}{d\varepsilon}\right|_{\varepsilon=0} = -\int_0^\infty K(t) H K(t) dt\,.
    \end{align}

    Combining this with \cref{eq:df1}, we get the second-order derivative of $f$
    \begin{align} \label{eq:df2}
        \left.\frac{d^2f}{d\varepsilon^2}\right|_{\varepsilon=0} = \frac12\Tr\left[\left.\frac{d\Delta}{d\varepsilon}\right|_{\varepsilon=0} H\right] = -\frac12\int_0^\infty\Tr[K(t) H K(t)H] dt \,.
    \end{align}
    
    To show that \cref{eq:df2} is negative, it suffices to show that $\Tr[K(t)HK(t)H]>0$ for any $t\geq 0$. We prove this by contradiction.

    Assume that the following holds for some $t$
    \begin{align} \label{eq:leqassumption}
        \Tr[K(t)HK(t)H]\leq 0 \,.
    \end{align}
    Take $X=HK(t)H = H^\dag K(t) H$, with the second equality due to $H$ being Hermitian. $X\geq0$ because $K(t)>0$. Since $K(t)>0$, if $X\geq 0$ and $X\neq 0$, one has $\Tr[K(t)X]>0$, contradicting the assumption. Therefore, $X=0$.

    Again, since $K(t)>0$, $X = H^\dag K(t) H = 0$ implies $H = 0$, in contradiction to $H\neq 0$ given at the beginning of this proof. Therefore, the assumption \cref{eq:leqassumption} is false and thus $\Tr[K(t)HK(t)H] >0$ for any $t\geq0$.

    We have obtained that $\left.\frac{d^2f}{d\varepsilon^2}\right|_{\varepsilon=0} < 0$ for any Hermitian operator $H\neq 0$. This proves that $F(\rho,\sigma)$ is strictly concave with respect to $\sigma$.
\end{proof}
\end{document}